\documentclass[runningheads]{llncs}
\usepackage{graphicx}
\usepackage{amsmath,amssymb,amsfonts}
\usepackage{multirow}
\usepackage{varwidth}

\usepackage[linesnumbered,ruled,vlined]{algorithm2e}
\newtheorem{assumption}{Assumption}
\usepackage[toc,page]{appendix}

\usepackage{caption}
\usepackage{subcaption}
\usepackage[table]{xcolor}
\usepackage{arydshln}

\usepackage[colorlinks=true,urlcolor=blue]{hyperref}

\usepackage{array}
\usepackage{booktabs}
\usepackage{graphicx}
\usepackage{threeparttable}
\usepackage{wasysym} 
\usepackage{multirow}

\usepackage[framemethod=tikz]{mdframed}

\makeatletter
\def\ProblemSpecBox{
  \@ifnextchar[\ProblemSpecBox@opt{\ProblemSpecBox@noopt}}
\def\ProblemSpecBox@opt[#1]#2{
  \protected@edef\@currentlabelname{#1}
  \protected@edef\@currentlabel{#1}
  \begin{mdframed}[
    innerleftmargin=5pt,
    innerrightmargin=5pt,
    innertopmargin = 5pt,
    innerbottommargin=5pt,
    skipabove=\dimexpr\topsep+\ht\strutbox\relax,
    roundcorner=2pt,
    frametitle={#2},
    frametitlerule=true,
    backgroundcolor=gray!7,
    frametitlebackgroundcolor=gray!20]
}
\def\ProblemSpecBox@noopt#1{
  \ProblemSpecBox@opt[#1]{#1}
}
\def\endProblemSpecBox{
  \end{mdframed}
}
\makeatother

\begin{document}
\title{Tight Short-Lived Signatures\thanks{This work appeared at IEEE Euro S\&P 2023 as Poster~\cite{arup_mondal_2023_8197649}.}}
%
%

\author{Arup Mondal\inst{1} \and Ruthu Hulikal Rooparaghunath\inst{2} \and Debayan Gupta\inst{1}}
%

%

\institute{Ashoka University, Sonipat, Haryana, India \\
\email{\{arup.mondal\_phd19, debayan.gupta\}@ashoka.edu.in}\\ \and
Vrije Universiteit,
Amsterdam, The Netherlands \\
\email{r.rooparaghunath@student.vu.nl}}

\maketitle              


\begin{abstract}
A Time-lock puzzle (TLP) sends information into the future: a predetermined number of sequential computations must occur (i.e., a predetermined amount of time must pass) to retrieve the information, regardless of parallelization.
Buoyed by the excitement around secure decentralized applications and cryptocurrencies, the last decade has witnessed numerous constructions of TLP variants and related applications (e.g., cost-efficient blockchain designs, randomness beacons, e-voting, etc.).

In this poster, we first extend the notion of TLP by formally defining the ``time-lock public key encryption'' (TLPKE) scheme. Next, we introduce and construct a ``tight short-lived signatures'' scheme using our TLPKE. Furthermore, to test the validity of our proposed schemes, we do a proof-of-concept implementation and run detailed simulations. 

\end{abstract}


\section{Introduction}\label{sec:intro}

A short-lived signature (SLS) provides a verifier with two possibilities: either a generated signature $\sigma$ on $m$ is correct, or a user has expended a minimum predetermined amount of sequential work ($T$ steps or time) to forge the signature. In other words, the signatures created remain valid for a short period of time $T$.
Formally, we define the unforgeability period as the time starting from when a signer creates a signature for a message using their private signing information (or key) and the \texttt{Sign} algorithm. Once the unforgeability period $T$ has elapsed, anyone can compute a forged signature using some public signing information and the \texttt{ForgeSign} algorithm.

Recall that any party can compute a forged signature (using \texttt{ForgeSign}) after the unforgeability period $T$ has passed. However, there is no guarantee that a party cannot generate a forged signature in \textit{advance}. To ensure this, we use the same model used in~\cite{DBLP:journals/iacr/ArunBC22}. The signature incorporates a random beacon value to ensure it was not created before a specific time $T_0$. Suppose a verifier observes the signature within $\hat{T}$ units of time after $T_0$. In this case, they will believe it is a valid signature if $\hat{T} < T$ because it would be impossible to have forged the signature within that time period. Once $\hat{T} \geq T$, the signature is no longer convincing as it may have been constructed through forgery.

\subsubsection*{Brief Concurrent Work}
Recently, Arun et al.~\cite{DBLP:journals/iacr/ArunBC22} studied the variants notion of short-lived cryptographic primitives, i.e., short-lived proofs and signatures. Similar to our work, they make use of sequentially-ordered computations ($T$-sequential computation) as a means to enforce time delay during which signatures are unforgeable but become forgeable afterward ($(1 + c) \cdot T$-sequential computations). 
In this work, we use the same models as used in~\cite{DBLP:journals/iacr/ArunBC22}, however, we define and construct \textbf{tight} short-lived signatures, where the forged signatures can be generated in time not much more than sequentially bound $T$. In other words, tight short-lived signatures ensure that forged signatures can be generated in \textbf{exactly} $T$ sequential computations. 

\begin{table}[ht]
\centering
\scriptsize
\caption{Complexity comparison of SLS schemes.}
\label{table:slscompare}
\begin{tabular}{l|c|c|c|c}
\hline
Paper & \shortstack{Setup \& Sign} & Forge Sign & Verify & Tight \\ \hline

Arun et al.~\cite{DBLP:journals/iacr/ArunBC22}      & $\textsf{poly}(\lambda)$     & $O((1+c) \cdot T)$ & VDF~(\cite{DBLP:journals/iacr/Pietrzak18a,DBLP:conf/eurocrypt/Wesolowski19})        & No          \\ \hline

Algorithm  & $\textsf{poly}(\lambda)$  & $O(T)$  & $O(1)$ & Yes           \\ \hline
\end{tabular}
\end{table}

\noindent
Short-lived cryptographic primitives have many real-life use cases; we refer to~\cite{DBLP:journals/iacr/ArunBC22} for a  detailed discussion of its applications.
\textit{Our main contributions are summarised as follows:}

\begin{itemize}
    \item First, we extend the time-lock puzzle~\cite{rivest1996time} by formally defining the ``time-lock public key encryption'' (TLPKE) scheme, and demonstrate a construction using a repeated squaring assumption in a group of unknown order (Sec.~\ref{sec:tlpke}).
    
    \item We introduce and construct a ``tight short-lived signature'' scheme from our TLPKE scheme (Sec.~\ref{sec:tsls}).
    
    \item We conduct a proof-of-concept implementation study and analyze the performance of our construction (Sec.~\ref{sec:tsls}).
\end{itemize}

\section{Technical Preliminaries}

\subsubsection*{Basic Notation}
Given a set $\mathcal{X}$, we denote by $x \overset{\$}{\leftarrow} \mathcal{X}$ the process of sampling a value $x$ from the uniform distribution on $\mathcal{X}$. Supp($\mathcal{X}$) denotes the support of the distribution $\mathcal{X}$. 
We denote by $\lambda \in \mathbb{N}$ the security parameter. A function \texttt{negl}: $\mathbb{N} \rightarrow \mathbb{R}$ is negligible if it is asymptotically smaller than any inverse-polynomial function, namely, for every constant $\epsilon > 0$ there exists an integer $N_{\epsilon}$ and for all $\lambda > N_{\epsilon}$ such that $\texttt{negl}(\lambda) \leq \lambda^{-\epsilon}$.

\subsubsection*{Number Theory}
We assume that $N=p \cdot q$ is the product of two large secret and \textit{safe} primes and $p \neq q$. We say that $N$ is a strong composite integer if $p = 2p' + 1$ and $q = 2q' +1$ are safe primes, where $p'$ and $q'$ are also prime. We say that $\mathbb{Z}_N$ consists of all integers in $[N]$ that are relatively prime to $N$ (i.e., $\mathbb{Z}_N = \{x \in \mathbb{Z}_N: \texttt{gcd}(x, N)=1\}$).

\subsubsection*{Repeated Squaring Assumption}
The repeated squaring assumption~\cite{rivest1996time} roughly says that there is no parallel algorithm that can perform $T$ squarings modulo an integer $N$ significantly faster than just doing so sequentially, assuming that $N$ cannot be factored efficiently, or in other words \texttt{RSW} assumption implies that factoring is hard. More formally, no adversary can factor an integer $N = p \cdot q$ where $p$ and $q$ are large secrets and ``safe'' primes (see \cite{DBLP:journals/iacr/Pietrzak18a} for details on ``safe'' primes). Repeated squaring \texttt{RSW = (Setup, Sample, Eval)} is defined below. Moreover, we define a trapdoor evaluation \texttt{RSW.tdEval} (which enables \textit{fast} repeated squaring evaluation), from which we can derive an actual output using trapdoor in $\textsf{poly}(\lambda)$ time.

\begin{mdframed}
\scriptsize
    $\bullet$ $N \leftarrow \texttt{RSW.Setup}(\lambda)$ : Output $\textsf{pp} = (N)$ where $N = p \cdot q$ as the product of two large ($\lambda$-bit) randomly chosen secret and safe primes $p$ and $q$.
    
    \noindent
    $\bullet$ $x \leftarrow \texttt{RSW.Sample}(\textsf{pp})$ : Sample a random instance $x$.
    
    \noindent
    $\bullet$ $y \leftarrow \texttt{RSW.Eval}(\textsf{pp}, T, x)$ : Output $y = x^{2^{T}} \mod N$ by computing the $T$ sequential repeated squaring from $x$. 
    
    \noindent
    $\bullet$ $y \leftarrow \texttt{RSW.tdEval}(\textsf{pp}, \textsf{sp} = \phi(N), x)$ : To compute $y = x^{2^{T}} \mod N$ \textit{efficiently} using the trapdoor as follows:
    \begin{enumerate}
        \item[-] Compute $v = 2^T \mod \phi(N)$. \\
        \textbf{Note:} $(2^T \mod \phi(N)) \ll 2^T$ for large $T$
            
        \item[-] Compute $y = x^v \mod N$. \\
        \textbf{Note:} $x^{2^{T}} \equiv x^{(2^T \mod \phi(N))} \equiv x^v \pmod{N}$.
    \end{enumerate}
\end{mdframed}

\begin{assumption}[$T$-Repeated Squaring Assumption without Trapdoor~\cite{rivest1996time}]\label{assump:assump1}
	For every security parameter $\lambda \in \mathbb{N}$, $N \in \emph{Supp}(\texttt{\emph{RSW.Setup}}(\lambda))$, $x \in \emph{Supp}(\texttt{\emph{RSW.Sample}}(N))$, and a time-bound parameter $T$, computing the $x^{2^T} \mod N$ without knowledge of a trapdoor or secret parameter \textsf{\emph{sp}} using the \texttt{\emph{RSW.Eval}} algorithm requires $T$-sequential time for algorithms with $\textsf{\emph{poly}}(\log(T), \lambda)$-parallel processors.
\end{assumption}

\section{Time-Lock Public Key Encryption}\label{sec:tlpke}

The notion of time-sensitive cryptography was introduced by Rivest, Shamir, and Wagner~\cite{rivest1996time} in 1996, in ``Time-lock puzzles and timed-release Crypto'' (TLP). They presented a construction using repeated squaring in a finite group of unknown order, resulting in an encryption scheme. This scheme allows the holder of a trapdoor to perform ``fast'' encryption or decryption, while others without the trapdoor can only do so slowly (requiring $T$ sequential computations).

For the purpose of our tight short-lived signature protocol, we require and define a variation of TLP. We follow the definitions given in~\cite{rivest1996time}, altered to fit the public key encryption paradigm, rather than symmetric key encryption.
This variation, which we refer to as ``Time-Lock Public Key Encryption'' (TLPKE)\footnote{TLP with public key encryption instead of symmetric key encryption}, can be described as a ``public key encryption scheme with sequential and computationally intensive derived private key generation''.

\subsubsection*{Protocol}
The formal details of our TLPKE construction from repeated squaring, is \texttt{TLPKE} = (\texttt{Setup}, \texttt{Eval}, \texttt{Encrypt}, \texttt{Decrypt}) specified in Algorithm~\ref{algo:tlpke}.

\begin{figure}
\begin{ProblemSpecBox}[1]{\small Algorithm 1: Time-Lock Public Key Encryption}
\scriptsize
\label{algo:tlpke}

        $\bullet$ \texttt{TLPKE.Setup}($\lambda, T$)
        \begin{enumerate}
            \item Call and generate $N \leftarrow \texttt{RSW.Setup}(\lambda)$
            
            \item Generate an input $x \in \mathbb{Z}_N^* \leftarrow \texttt{RSW.Sample}(\textsf{pp})$ 
            
            \item Generates a key pair $(pk, sk)$ for a semantically secure public-key encryption scheme: $\texttt{PKE} = (\texttt{GenKey, Enc, Dec})$.
            
            \item Encrypt the $sk$ as $ek = sk + x^{2^T} \mod N$
            
            \item Compute $y = x^{2^T} \mod N \in \mathbb{Z}_N^*$ \textit{efficiently} using the trapdoor evaluation $\texttt{RSW.tdEval}(\textsf{pp}, \textsf{sp} = \phi(N), x)$.
            
            \item \textbf{return} $\textsf{pp} = (N, T, x, pk, ek)$.
        \end{enumerate}
        
        \noindent
        $\bullet$ \texttt{TLPKE.Eval}(\textsf{pp})
        \begin{enumerate}
            \item Compute $y = x^{2^T} \mod N \in \mathbb{Z}_N^*$ using $\texttt{RSW.Eval}(\textsf{pp}, T, x)$
            \item Extract the decryption key $sk = ek - y$
            \item \textbf{return} $(y, sk)$ 
        \end{enumerate}
        
        \noindent
        $\bullet$ \texttt{TLPKE.Encrypt}($\textsf{pp}, M$)
        \begin{enumerate}
            \item Encrypt a message $M$ with key $pk$ and a standard encryption \texttt{Enc}, to obtain the ciphertext $C_M = \texttt{Enc}(pk, M \parallel x)$ and \textbf{return} $C_M$.
        \end{enumerate}
        
        \noindent
        $\bullet$ \texttt{TLPKE.Decrypt}($\textsf{pp}, sk, y, C_M$)
        \begin{enumerate}
            \item Decrypt the message as $M \parallel x = \texttt{Dec}(sk, C_M)$
            \item Parse $M \parallel x$ and \textbf{return} the message $M$
        \end{enumerate}
\end{ProblemSpecBox}

\end{figure}

\section{Tight Short-Lived Signature}\label{sec:tsls}

\subsubsection*{Syntax and Security Definitions}
Here, we recall and modify the definition of short-lived signatures (SLS) from~\cite{DBLP:journals/iacr/ArunBC22} and define our tight SLS as follows:

\begin{definition}[Tight Short-Lived Signatures]\label{def:sls}
Let $\lambda \in \mathbb{N}$ be a security parameter and a space of random beacon $\mathcal{R} \geq 2^{\lambda}$. A short-lived signature \texttt{\emph{SLS}} is a tuple of four probabilistic polynomial time algorithms \texttt{\emph{(Setup, Sign, ForgeSign, Verify)}}, as follows:
\begin{itemize}
    \item[$\bullet$] $\texttt{\emph{Setup}}(\lambda) \rightarrow (\textsf{\emph{pp}}, sk)$, is randomized algorithm that takes a security parameter $\lambda$ and outputs public parameters $\textsf{\emph{pp}}$ and a secret key $sk$ (the $sk$ can \textbf{only} be accessed by the \texttt{\emph{SLS.Sign}} algorithm). The public parameter $\textsf{\emph{pp}}$ contains an input domain $\mathcal{X}$, an output domain $\mathcal{Y}$, and time-bound parameter $T$.
    
    \smallskip
    \item[$\bullet$] $\texttt{\emph{Sign}}(\textsf{\emph{pp}}, m, r, sk) \rightarrow \sigma$, takes a public parameter $\textsf{\emph{pp}}$, a secret parameter $sp$, a message $m$ and a random beacon $r$, and outputs (in time less than the predefined time bound $T$) a signature $\sigma$.
    
    \smallskip
    \item[$\bullet$] $\texttt{\emph{ForgeSign}}(\textsf{\emph{pp}}, m, r) \rightarrow \sigma$, takes a public parameter $\textsf{\emph{pp}}$, a message $m$ and a random beacon $r$, and outputs (in time \textbf{exactly} $T$) a signature $\sigma$.
    
    \smallskip
    \item[$\bullet$] $\texttt{\emph{Verify}}(\textsf{\emph{pp}}, m, r, \sigma) \rightarrow \{\textsf{\emph{accept}}, \textsf{\emph{reject}}\}$, is a deterministic algorithm takes a public parameter $\textsf{\emph{pp}}$, a message $m$ and a random beacon $r$ and a signature $\sigma$, and outputs \textsf{\emph{accept}} if $\sigma$ is the correct signature on $m$ and $r$, otherwise outputs \textsf{\emph{reject}}.
\end{itemize}
\end{definition}

\noindent
A $\texttt{SLS}$ must satisfy the three properties \textbf{Correctness} (Definition~\ref{definition:slscorrect}), \textbf{Existential Unforgeability} (Definition~\ref{definition:slsexunforg}), and \textbf{Indistinguishability} (Definition~\ref{definition:slsindisnt}) as follows:

\begin{definition}[Correctness]\label{definition:slscorrect}
A \texttt{\emph{SLS}} is correct (or complete) if for all $\lambda \in \mathbb{N}$, $m$, and $r \in \mathcal{R}$  it holds that,
\begin{equation*}
    \scriptsize
    \emph{Pr} \left[ 
    \begin{aligned}
        \texttt{\emph{Verify}}(\textsf{\emph{pp}},m,r,\sigma) = \textsf{\emph{accept}}    
    \end{aligned}
    \middle \vert 
    \begin{aligned}
        & \texttt{\emph{Setup}}(\lambda) \rightarrow (\textsf{\emph{pp}}, sk) \\
        & \texttt{\emph{Sign}}(m,r,sk) \rightarrow \sigma
    \end{aligned}
    \right]
    =1
\end{equation*}
\end{definition}

\begin{definition}[$T$-Time Existential Unforgeability]\label{definition:slsexunforg}
A \texttt{\emph{SLS}} has $T$-time existential unforgeability if $\forall$ $\lambda, T \in \mathbb{N}$, $m$ and $r \in \mathcal{R}$, and all pairs of PPT algorithms $(\mathcal{A}, \mathcal{A'})$, such an $\mathcal{A}$ (offline) can run in total time $\textsf{\emph{poly}}(T, \lambda)$ and in a parallel running time of $\mathcal{A'}$ (online) on at most \textsf{\emph{poly}}($\log T, \lambda$)-processors is less than $T$, there exists a negligible function \texttt{\emph{negl}} such that,
\begin{equation*}
    \scriptsize
    \emph{Pr} \left[ 
    \begin{aligned}
        & \texttt{\emph{Sign}}(m, r, sk) \rightarrow \sigma \\
        & \texttt{\emph{Verify}}(\textsf{\emph{pp}},m^*,r,\sigma^*) \\ & =    \textsf{\emph{accept}}   
    \end{aligned}
    \middle \vert
    \begin{aligned}
        & \texttt{\emph{Setup}}(\lambda) \rightarrow (\textsf{\emph{pp}}, sk) \\
        & \mathcal{A}(\textsf{\emph{pp}}, \lambda, T) \rightarrow \alpha \\
        & \mathcal{A'}(\textsf{\emph{pp}}, m^*, r, \alpha) \rightarrow \sigma^*
    \end{aligned}
    \right]
    \leq\texttt{\emph{negl}}(\lambda)
\end{equation*}

\end{definition}

\begin{definition}[Indistinguishability]\label{definition:slsindisnt}
A \texttt{\emph{SLS}} is computationally indistinguishable (and statistically indistinguishable; when taken over the random coins used by each algorithm and randomly generated private parameters) if for all $\lambda \in \mathbb{N}$, $m$, and $r \in \mathcal{R}$  it holds that,
\begin{equation*}
    \scriptsize
        \left \vert
        \begin{aligned}
        \emph{Pr} \left[
        \begin{aligned}
            \mathcal{A}(\textsf{\emph{pp}}, m, r, \sigma) = \textsf{\emph{accept}} \\
        \end{aligned}
        \middle \vert 
        \begin{aligned}
            & \texttt{\emph{Setup}}(\lambda) \rightarrow (\textsf{\emph{pp}}, sk) \\
            & \texttt{\emph{Sign}}(m, r, sk) \rightarrow \sigma \\
        \end{aligned}
        \right] - \\
        \emph{Pr} \left[
        \begin{aligned}
            \mathcal{A}(\textsf{\emph{pp}}, m, r, \sigma) = \textsf{\emph{accept}}
        \end{aligned}
        \middle \vert 
        \begin{aligned}
            & \texttt{\emph{Setup}}(\lambda) \rightarrow (\textsf{\emph{pp}}, sk) \\
            & \texttt{\emph{ForgeSign}}(\textsf{\emph{pp}}, m, r) \rightarrow \sigma \\
        \end{aligned}
        \right] \\
        \end{aligned}
        \right \vert
        \leq \texttt{\emph{negl}}(\lambda)
\end{equation*}
\end{definition}

\noindent
Our scheme, formalized in Definition~\ref{def:sls}, presents an efficient generalized framework for short-lived signatures (Algorithm~\ref{algorithm:sls}) that is compatible with all signature schemes. However, note that while our Indistinguishability definition (Definition~\ref{definition:slsindisnt}) compares distributions of output, some signature schemes are deterministic (e.g., BLS, RSA signature). In such cases, it is necessary for \texttt{Sign} and \texttt{Forge} to produce the exact signature (e.g., Schnorr signature) with overwhelming probability.

\subsubsection*{Protocol Design}

The formal construction of tight short-lived signatures \texttt{SLS} = (\texttt{Setup}, \texttt{Sign}, \texttt{ForgeSign}, \texttt{Verify}) using our TLPKE is specified in Algorithm~\ref{algorithm:sls}.

\begin{ProblemSpecBox}[2]{\small Algorithm 2: Tight Short-Lived Signatures from TLPKE}
\label{algorithm:sls}
\scriptsize
    $\bullet$ $\texttt{SLS.Setup}(\lambda)$
    \begin{enumerate}
        \item Call and generate $N \leftarrow \texttt{RSW.Setup}(\lambda)$
            
        \item Generate an input $x \in \mathbb{Z}_N^* \leftarrow \texttt{RSW.Sample}(\textsf{pp})$.

        \item Choose a time bound parameter $T \in T(\lambda)$.
            
        \item Generates a key pair $(pk, sk) = \mathrm{\Pi}_{\texttt{KeyGen}}(\lambda)$ for a Signature scheme: $\mathrm{\Pi} = (\texttt{KeyGen, Sign, Verify})$.
            
        \item Compute $y = x^{2^T} \mod N \in \mathbb{Z}_N^*$ \textit{efficiently} using the trapdoor evaluation $\texttt{RSW.tdEval}(\textsf{pp}, \textsf{sp} = \phi(N), x)$.
        
        \item Encrypts the $sk$ as $ek = sk + x^{2^T} \mod N$
        
        \item Output public parameter $\textsf{pp} = (N, T, x, pk, ek)$\footnote{\scriptsize The \textsf{pp} can be generate by calling $\texttt{TLPKE.Setup}(\lambda)$.} and secret key $sk$ generated by $\mathrm{\Pi}$ ($sk$ can only be accessed by the \texttt{SLS.Sign}).
    \end{enumerate}
    
    \noindent
    $\bullet$ $\texttt{SLS.Sign}(m, r, sk)$
    \begin{enumerate}
        \item Compute $M = \mathcal{H}(m \parallel r)$.
        \item Compute a signature $\sigma = \mathrm{\Pi}_{\texttt{Sign}}(sk, M)$.
        \item Output a short-lived signature ($\sigma, r$).
    \end{enumerate}
    
    \noindent
    $\bullet$ $\texttt{SLS.ForgeSign}(\textsf{pp}, m, r)$
    \begin{enumerate}
        \item Compute $M = \mathcal{H}(m \parallel r)$.
        \item Call and extract $(y, sk) = \texttt{TLPKE.Eval}(\textsf{pp})$. 
        \item Compute a forge signature $\sigma = \mathrm{\Pi}_{\texttt{Sign}}(sk, M)$\footnote{\scriptsize Forge signature $\sigma$ can be computed in time not much more than the sequentiality bound \textit{\textbf{exactly}} $T$ even on a parallel computer with $\textsf{poly}(\log T, \lambda)$ processors. Therefore, our SLS is Tight Short-Lived Signature.}.
        \item Output a forge short-lived signature $(\sigma, r)$
    \end{enumerate}
    
    \noindent
    $\bullet$ $\texttt{SLS.Verify}(\textsf{pp}, m, r, \sigma)$
    \begin{enumerate}
        \item Compute $M = \mathcal{H}(m \parallel r)$.
        \item Check that $\mathrm{\Pi}_{\texttt{Verify}}(pk, M, \sigma)$\footnote{\scriptsize The signature verification algorithm can be computed in $O(1)$ time.}.
    \end{enumerate}
\end{ProblemSpecBox}

\begin{theorem}[Tight Short-Lived Signatures]\label{theorem:slsde}
Assuming that $\mathcal{H}$ is a random oracle, \texttt{\emph{RSW}} is the repeated point squaring assumption, and \texttt{\emph{TLPKE}} is a time-lock public key encryption scheme (see \emph{Algorithm~\ref{algo:tlpke}}), it holds that the protocol \texttt{\emph{SLS}} (\emph{Algorithm}~\ref{algorithm:sls}) is a tight short-lived signature scheme.
\end{theorem}

\begin{proof}[Proof Sketch]
The correctness of the \texttt{SLS} scheme is proven by the correctness of the underlying time-lock public key encryption \texttt{TLPKE}.
Indistinguishability is trivial as the signing and forgery produce the exact signature using the underlying signature scheme, given that the \texttt{TLPKE.Eval} (key extraction) and \texttt{TLPKE.Decrypt} (decryption) algorithms of the underlying \texttt{TLPKE} are deterministic. 
The $T$-Time Existential Unforgeability is a direct result of the sequentiality and security ($T$-IND-CPA Security) property of the underlying \texttt{TLPKE} and modeling $\mathcal{H}$ as a random oracle.
\end{proof}

\subsubsection*{Experimental Results}
We use \textsf{Python} to implement \texttt{RSW} primitive (and hence our proposed TLPKE and tight SLS). 
The experiments are performed using a Windows 11 system with Intel(R) Core(TM) i5-1035G1 CPU @1.00GHz with 8 GB RAM. 
Note that \texttt{RSW} is the underlying required primitive of our repeated squaring-based TLPKE and tight SLS. Hence, in this poster, we do not provide detailed simulation results for TLPKE and tight SLS, focusing instead on profiling the underlying workhorse primitive.

In Figure~\ref{fig:rsw}, we show the experimental results for \texttt{RSW} evaluation. The \texttt{RSW.tdEval} (Figure~\ref{fig:rswtdeval}) run time changes linearly with the security parameter $\lambda$ (the bit length of $N$ is derived from the bit length of $\lambda$). As shown in Figure~\ref{fig:rsweval}, the time taken to compute the \texttt{RSW.Eval} increases with an increase in the number of exponentiations. Changes in  time $T$ yield a great variation in the evaluation time.

\begin{figure}[ht]
     \centering
     \begin{subfigure}[b]{0.45\textwidth}
         \centering
         \includegraphics[width=\textwidth]{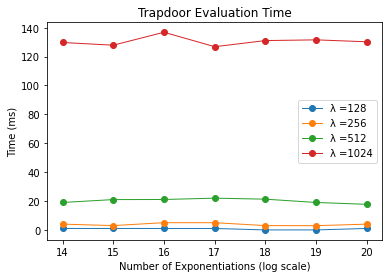}
         \caption{\scriptsize Trapdoor Evaluation Time.}
         \label{fig:rswtdeval}
     \end{subfigure}
     \hfill
     \begin{subfigure}[b]{0.45\textwidth}
         \centering
         \includegraphics[width=\textwidth]{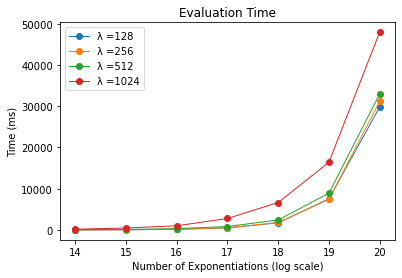}
         \caption{\scriptsize Evaluation Time.}
         \label{fig:rsweval}
     \end{subfigure}
     \caption{Trapdoor Evaluation and $T$-Sequential Evaluation of the Repeated Squaring \texttt{RSW}. $j$ is labelled as ``Number of Exponentiations'', $T = 2^j$.}
     \label{fig:rsw}
\end{figure}

\section{Future Work}

We conclude with an open problem: Arun et al.~\cite{DBLP:journals/iacr/ArunBC22} define \textit{reusable forgeability} property in the context of short-lived proofs (see Sec.~4.1 in~\cite{DBLP:journals/iacr/ArunBC22}), which ensure that one slow computation for a random beacon value (say $r$) enables efficiently forging a proof for any statement (say $x$) without performing a full additional slow computation. Furthermore, Arun et al.~\cite{DBLP:journals/iacr/ArunBC22} extend reusable forgeability in the context of short-lived signatures and describe a construction (see Sec.~8.3 in~\cite{DBLP:journals/iacr/ArunBC22}). \textit{In the near future, we hope to construct an efficient \textbf{tight} reusable and forgeable short-lived signature scheme.}


\bibliographystyle{splncs04}
\bibliography{arxiv-ref}

\end{document}